\documentclass{amsart}
\usepackage{amsmath}
\usepackage{amssymb,amsthm}
\usepackage{graphicx}
\usepackage{multirow}
\usepackage{float}
\usepackage{supertabular}
\numberwithin{equation}{section}
\newtheorem{thm}{Theorem}[section]
\newtheorem{rem}[thm]{Remark}

\title[]{Lower bounds to eigenvalues of one-electron Hamiltonians}
\author{Sohei Ashida}
\begin{document}
\maketitle

\begin{abstract}
A method for computing lower bounds to eigenvalues of sums of lower semibounded self-adjoint operators is presented. We apply the method to one-electron Hamiltonians. To improve the lower bounds we consider symmetry of molecules and use Temple's inequality. These methods would be useful in estimating eigenvalues of electronic Hamiltonians needed for studies of properties of molecules.
\end{abstract}

\section{Introduction}\label{firstsec}
In studies of properties of molecules, hypersurfaces called potential energy surfaces would be the most important subject. The potential energy surface is a graph of a certain function of positions $x_1,\dots,x_n\in \mathbb R^3$ of nuclei. The value of the function equals to an eigenvalue of a differential operator called electronic Hamiltonian. Therefore, for the studies of chemical properties of molecules accurate estimates for eigenvalues of electronic Hamiltonians are needed. The simplest kind of electronic Hamiltonians is one-electron Hamiltonian which is also important for estimating eigenvalues of electronic Hamiltonians with more than one electrons (cf. \cite{GHT,BG,Ga,BF2,Wi}). One-electron Hamiltonians are operators on $L^2(\mathbb R^3)$ and written as follows
\begin{equation}\label{myeq3.0}
A:=-\frac{1}{2}\Delta_y-\sum_{k=1}^n\frac{Z_k}{\lvert y-x_k\rvert},
\end{equation}
where $y\in\mathbb R^3$ is a position of an electron, the parameters $x_1,\dots,x_n\in\mathbb R^3$ are positions of nuclei and $Z_1,\dots,Z_n\in\mathbb N$ are atomic numbers of the nuclei.
The domain of $A$ is $D(A)=H^2(\mathbb R^{3})$ and $A$ is a lower semibounded self-adjoint operator (cf. \cite[Theorems X.12 and X.15]{RS}).  Upper bounds to eigenvalues of $A$ are obtained by variational methods. In this paper we develop methods to give lower bounds. 

For lower bounds we regard $A$ as a sum operators $A_k$:
$$A:=\sum_{k=1}^n A_k,$$
where
\begin{equation}\label{myeq1.1}
A_k:=-\frac{1}{2n}\Delta_y-\frac{Z_k}{\lvert y-x_k\rvert}.
\end{equation}

Using spectral projections of $A_k$ we estimate the operator $A$ from below by a sum of a finite dimensional operator $\tilde A$ and a constant. The eigenvalue problem of $\tilde A$ is reduced to a matrix eigenvalue problem. To improve the lower bound we use Temple's inequality. Because we need a lower bound to the second eigenvalue of $A$ greater than a variational upper bound to the first eigenvalue for Temple's inequality to be available, we improve the lower bound for the second eigenvalue considering symmetry of the molecule.

Section \ref{sec2} gives a lower bound procedure for sums of lower semibounded self-adjoint operators. In Section \ref{thirdsec} the method is applied to one electron Hamiltonians. Section \ref{fifthsec} describes improvements of the lower bounds by Temple's inequality and eigenvalue problems adapted to symmetry of molecules.

\section{Lower bound procedures}\label{sec2}
Let $\mathcal H$ be a separable complex Hilbert space with norm $\lVert u\rVert$ and inner product $\langle u,v\rangle$ (conjugate linear with respect to the first argument). Let $n$ be a natural number and  $A_k,\ k=1,\dots,n$ be self-adjoint operators with domains $D(A_k)$ dense in $\mathcal H$. We assume $A_k$ are lower semibounded, i.e. there exist $C_k>0$ such that $\langle A_ku,u\rangle\geq -C_k\lVert u\rVert^2$ for any $u\in D(A_k)$. We denote the essential spectra of $A_k$ by $\sigma_{ess}(A_k)$ and set $\lambda_{\infty}^k:=\inf\sigma_{ess}(A_k)$. We denote the lowest-ordered eigenvalues of $A_k$ by
$$\lambda^k_1\leq\lambda^k_2\leq\dotsm\leq\lambda^k_{\infty},$$
and the associated orthonormal eigenvectors by $\psi^k_1,\psi^k_2,\dotsm$.

Let $\sum_{k=1}^nA_k$ be essentially self-adjoint on $\bigcap_{k=1}^nD(A_k)$ and denote the self-adjoint extension by $A$. Set $\mu_{\infty}:=\inf\sigma_{ess}(A)$ and let
$$\mu_{1}\leq\mu_{2}\leq\dotsm\leq\mu_{\infty}$$
be eigenvalues of $A$.
Set $\lambda<\min\{\lambda_{\infty}^1,\lambda^{2}_{\infty},\dots,\lambda^n_{\infty}\}$. We label all $(k,l)$ satisfying $\lambda^k_l\leq\lambda$ by $i=1,\dots,M\in\mathbb N$ and write $\lambda^k_l=\lambda_i$, $\psi^k_l=\psi_i$. We denote $(k,l)$ corresponding to $i$ by $(k(i),l(i))$.  Set $\tilde \lambda_k:=\min\{\lambda^k_l: \lambda^k_l>\lambda\}$. We define a  matrix $\mathcal B$ by
$$\mathcal B_{ij}:=(\lambda_i-\tilde\lambda_{k(i)})\langle \psi_{i},\psi_j\rangle.$$
Let $\hat\mu_{1}\leq\hat\mu_{2}\leq\dotsm\leq\hat\mu_{M}$ be eigenvalues of $\mathcal B$.
\begin{thm}\label{lb}
 Assume $\psi_i$ are linearly independent. Then $\mathcal B$ is similar to a diagonal matrix and $\hat \mu_l\in\mathbb R,\ l =1,2,\dots,M$. Moreover, we have the following lower bounds for $\mu_i$.
$$\mu_l\geq\hat \mu_l+\sum_{k=1}^n\tilde\lambda_k,\ l=1,2,\dots,M.$$
\end{thm}
\begin{proof}
Let $E_k(\lambda)$ be the resolution of identity of $A_k$. Then the operator $A_k$ is decomposed as
$$A_k=E_k(\lambda)A_k+(1-E_k(\lambda))A_k.$$
Since on $\mathrm{Ran}\, (1-E_k(\lambda))$ the infimum of the spectra of $A_k$ is $\tilde \lambda_k$, the second term is estimated as
$$(1-E_k(\lambda))A_k\geq\tilde \lambda_k(1-E_k(\lambda)).$$
Therefore, we find
\begin{equation}\label{myeq2.1}
A_k\geq E_k(\lambda)A_k+(1-E_k(\lambda))\tilde\lambda_k=E_k(\lambda)(A_k-\tilde \lambda_k)+\tilde\lambda_k.
\end{equation}
Defining an operator $\tilde A$ by
$$\tilde A:=\sum_{k=1}^nE_k(\lambda)(A_k-\tilde \lambda_k),$$
\eqref{myeq2.1} gives
$$A\geq \tilde A+\sum_{k=1}^n\tilde\lambda_k.$$
Set $\tilde\mu_{\infty}:=\inf \sigma_{ess}(\tilde A)$ and let $\tilde\mu_{1}\leq\tilde\mu_{2}\leq\dotsm\leq\tilde\mu_{\infty}$ be eigenvalues of $\tilde A$.
Then we have
$$\mu_i\geq\tilde \mu_i+\sum_{k=1}^n\tilde \lambda_k,\ i=1,2,\dots.$$

Let us show
$\hat\mu_l=\tilde\mu_l,\ l=1,\dots,M.$
By min-max principle (cf. \cite[Theorem XIII.2]{RS}) $\tilde \mu_l$ is expressed as
\begin{equation}\label{myeq2.2}
\tilde \mu_l=\sup_{\varphi_1,\dotsm,\varphi_{l-1}\in\sum_{k=1}^nE_k(\lambda)\mathcal H}U_{\tilde A}(\varphi_1,\dotsm,\varphi_{l-1}),\ ,\ l=1,\dots,M,
\end{equation}
$$U_{\tilde A}(\varphi_1,\dotsm,\varphi_{l-1}):=\inf_{\substack{\psi\in\sum_{k=1}^nE_k(\lambda)\mathcal H,\ \lVert \psi\rVert=1 \\ \psi\in[\varphi_1\dotsm\varphi_{l-1}]^{\perp}}}\langle\psi,\tilde A\psi\rangle,$$
where $[\varphi_1\dotsm\varphi_{l-1}]^{\perp}$ is shorthand for $\{\psi: \langle \psi,\varphi_j\rangle=0,\ j=1,2,\dots,l-1\}$.
Here we note any $\psi\in \sum_{k=1}^nE_k(\lambda)\mathcal H$ can be written as $\psi=\sum_{i=1}^M c_i\psi_i$. Substituting this expression gives
$$U_{\tilde A}(\varphi_1,\dotsm,\varphi_{l-1})=\inf_{\substack{(\mathbf c,\mathcal G\mathbf c)=1\\ (\mathbf c,\mathcal G\mathbf{b}_j)=0,\ j=1,\dots,l-1}}(\mathbf c,\mathcal A\mathbf c),$$
where $\mathbf c:=^t(c_1,\dots,c_M)$, $\mathcal G_{ij}:=\langle \psi_i,\psi_j\rangle$, $\mathbf b_j=^t((b_j)_1,\dots,(b_j)_M)$ is defined by $\varphi_j=\sum_{i=1}^M( b_j)_i\psi_i$ and
$$\mathcal A_{ij}:=\sum_{s=1}^M(\lambda_s-\tilde\lambda_{k(s)})\langle\psi_i,\psi_s\rangle\langle\psi_s,\psi_j\rangle.$$
Here, $(\mathbf u,\mathbf v)$ is the usual inner product in $\mathbb C^M$.
The matrices $\mathcal A$ and $\mathcal G$ are Hermitian. Since $\psi_i$ are linearly independent, $\mathcal G$ is regular and $\mathcal G>0$. By a change of  bases in $\mathbb C^M$
\begin{equation}\label{myeq2.3}
U_{\tilde A}(\varphi_1,\dotsm,\varphi_{l-1})=\inf_{\substack{\lVert \tilde{\mathbf c}\rVert=1\\ (\tilde {\mathbf c},\tilde{\mathbf b_j})=0,\ j=1,\dots,l-1}}(\tilde{\mathbf c},\mathcal G^{-1/2}\mathcal A\mathcal G^{-1/2}\tilde{\mathbf c}),
\end{equation}
where $\tilde{\mathbf c}:=\mathcal G^{1/2}\mathbf c,\ \tilde{\mathbf b}_j:=\mathcal G^{1/2}\mathbf b_j$. By \eqref{myeq2.2} and \eqref{myeq2.3} we have
$$\tilde \mu_l=\sup_{\tilde{\mathbf b}_1\dots,\tilde{\mathbf b}_{l-1}}\inf_{\substack{\lVert \tilde{\mathbf c}\rVert=1\\ (\tilde{\mathbf c},\tilde{\mathbf b}_j)=0,\ j=1,\dots,l-1}}(\tilde{\mathbf c},\mathcal G^{-1/2}\mathcal A\mathcal G^{-1/2}\tilde{\mathbf c}),\ l=1,\dots,M.$$
The right-hand side is the expression by min-max principle for eigenvalues of $\mathcal G^{-1/2}\mathcal A\mathcal G^{-1/2}$. Thus $\tilde \mu_l,\ l=1,\dots,M$ are eigenvalues of $\mathcal G^{-1/2}\mathcal A\mathcal G^{-1/2}$. By the relation
$$\mathcal G^{-1/2}\mathcal G^{-1/2}\mathcal A\mathcal G^{-1/2}\mathcal G^{1/2}=\mathcal G^{-1}\mathcal A,$$
we can see that $\mathcal G^{-1}\mathcal A$ is similar to $\mathcal G^{-1/2}\mathcal A\mathcal G^{-1/2}$ and the eigenvalues of $\mathcal G^{-1}\mathcal A$ are $\tilde \mu_l,\ l=1,\dots,M$. Since $\mathcal G^{-1/2}\mathcal A\mathcal G^{-1/2}$ is Hermitian, $\tilde \mu_l$ are real and $\mathcal G^{-1}\mathcal A$ is similar to a diagonal matrix. Thus, if we show $\mathcal G^{-1}\mathcal A=\mathcal B$, we know that $\hat\mu_l=\tilde\mu_l$ and the theorem is proved.
Setting $\mathcal F:=\mathcal G^{-1}$, we have the obvious relation $\sum_s\mathcal F_{is}\mathcal G_{sj}=\delta_{ij}$ and $\mathcal A_{ij}=\sum_s(\lambda_s-\tilde\lambda_{k(s)})\mathcal G_{is}\mathcal G_{sj}$, where $\delta_{ij}$ is the Kronecker delta. These relations yield
$$(\mathcal G^{-1}\mathcal A)_{ij}=\sum_{s,t}(\lambda_s-\tilde\lambda_{k(s)})\mathcal F_{it}\mathcal G_{ts}\mathcal G_{sj}=(\lambda_i-\tilde\lambda_{k(i)})\mathcal G_{ij}=\mathcal B_{ij},$$
which completes the proof.
\end{proof}

\begin{rem}
The same result as the latter half of the theorem is used in \cite[Section 6]{GHT} to obtain lower bounds to eigenvalues of the electronic Hamiltonian of $\mathrm{H}_2^+$. Essentially the same result as the latter half has been obtained also by \cite{BF}, though the result is not formulated as bounds by matrix eigenvalues. The first half of the statement is new as far as the author knows.
\end{rem}

\section{Lower bounds for one-electron Hamiltonians}\label{thirdsec}
We apply Theorem \ref{lb} to one-electron Hamiltonians \eqref{myeq3.0}.
 The eigenvalues of $A_k$ in \eqref{myeq1.1} are $\lambda_j^k=-n/2j^2$. Since $\lambda_j^k$ are independent of $k$ we denote it by $\lambda_j$. The multiplicity of $\lambda_j$ is $j^2$. Since by dilation $U_n\psi(x):=n^{-3/2}\psi(nx)$, $A_k$ is transformed as 
\begin{align*}
&U_n^{-1}A_kU_n=n\left(-\frac{1}{2}\Delta_y-\frac{Z_k}{\lvert y-nx_k\rvert}\right)=n\tilde A_k,\\
&\tilde A_k:=-\frac{1}{2}\Delta_y-\frac{Z_k}{\lvert y-nx_k\rvert},
\end{align*}
the inner products of eigenfunctions of $A_k$ are the same as those of corresponding eigenfunctions of $\tilde A_k$ and eigenvalues of $A_k$ are $n$ times those of $\tilde A_k$.
The eigenvalue problem for $\tilde A_k$ is the one for Hydrogenlike-atoms and solvable (cf. \cite{EWK}). The eigenfunctions of $\tilde A_k$ are
\begin{equation}\label{myeq3.1}
\begin{split}
&\psi_{jlm}^k(y):=R_{jl}(Z_k,r_k)Y_{l,m}(\theta_k,\varphi_k),\\
&l=0,1,\dots,,j-1,\ m=-l,-l+1,\dots,l-1,l.
\end{split}
\end{equation}
where
$$R_{jl}(Z,r):=\sqrt{\frac{4(j-l-1)!}{j^4[(j+l)!]^3}}Z^{3/2}\rho^le^{-\rho/2}L_{j+l}^{2l+1}(\rho),\ \rho:=\frac{2Z}{j}r.$$
Here $(r_k,\theta_k,\varphi_k)$ is the polar coordinates of $y-nx_k$, $L_{j+l}^{2l+1}$ represents the associated Laguerre polynomial and  $Y_{lm}$ denotes the spherical harmonics
$$Y_{l,m}(\theta,\varphi):=\frac{1}{2}\sqrt{\frac{(2l+1)(l-\lvert m\rvert)!}{\pi(l+\lvert m\rvert)}}P_l^{\lvert m\rvert}(\cos\theta)e^{im\varphi},$$
where $P_l^{\lvert m\rvert}$ represents the associated Legendre polynomial. Since real functions are more convenient for the calculations of inner products of eigenfunctions, we use linear combinations of $Y_{l,m}$ and $Y_{l,-m}$ instead of $Y_{l,m}$ (cf. \cite{EWK}).

Let a function $\psi^k$ have the same form as in \eqref{myeq3.1}. For the calculation of inner product as $\langle\psi^k,\psi^l\rangle$ the origin of coordinates is placed at the center of the nuclei $k$ and $l$, and the direction of the $z$-axis is chosen to be the same as that of the line connecting the nuclei $k$ and $l$. We introduce spheroidal coordinates $\xi:=\frac{r_k+r_l}{R}$, $\eta:=\frac{r_k-r_l}{R}$, $\varphi:=\mathrm{arccos}(x/\sqrt{x^2+y^2})$, where $r_k$ and $r_l$ are the distances of the point $(x,y,z)$ from the nuclei $k$ and $l$ respectively, and $R$ is the distance between the nuclei $k$ and $l$ (cf. \cite{Ru}, \cite{Sl}). Then by the change of variables for any $f\in L^2(\mathbb R^3)$ we have
$$\int_{\mathbb R^3} f(x,y,z)dxdydz=\frac{R^3}{8}\int_1^{\infty}\left\{\int_{-1}^1\left\{\int_0^{2\pi}(\xi^2-\eta^2)f(\xi,\eta,\varphi)d\varphi\right\}d\eta\right\}d\xi.$$
By the change of variables we can calculate $\langle\psi^k,\psi^l\rangle$ as a function of $R$. Therefore, we can calculate the matrix elements of $\mathcal B$ in Theorem \ref{lb}.

The lower bounds to the first eigenvalues of \eqref{myeq3.0} with $n=2$, $Z_1=Z_2=1$ (i.e. the case of Hydrogen molecule ion $\mathrm H_2^+$) are listed in Table \ref{mytab}. 
The parameter $R$ indicates the distance between the tow nuclei. The lower bounds are calculated setting $\lambda\in [\lambda_j,\lambda_{j+1}),\ j=1,2,3$, where $\lambda$ is the real number which appears in the definition of $\mathcal B$ above the Theorem \ref{lb}. We also calculated the variational upper bounds $\mu_{1UB}$ to the eigenvalues using as trial functions
\begin{equation}\label{myeq3.2}
\psi=\exp (-\alpha R\xi/2)(1+\beta R^2\eta^2/4),
\end{equation}
and optimizing $\alpha$ and $\beta$ as in \cite{GHT}. Lower bounds $\mu_{1LB}$ by consideration of symmetry and Temple's inequality is explained in Section \ref{fifthsec}.
\begin{table}[H]
  \begin{center}
    \caption{The first eigenvalue bounds for $\mathrm{H}_2^+$}
    \setlength{\tabcolsep}{5pt}
    \begin{tabular}{cccccc} \hline
    \multicolumn{1}{c}{\multirow{2}{*}{R}}&\multicolumn{3}{c}{\begin{tabular}{c}Lower  Bounds\\  with $\lambda\in[\lambda_j,\lambda_{j+1})$\end{tabular}}&\multicolumn{1}{c}{\multirow{2}{*}{$\mu_{1UB}$}}&\multicolumn{1}{c}{\multirow{2}{*}{$\mu_{1LB}$}}\\
	    &\multicolumn{1}{c}{$j=1$}&\multicolumn{1}{c}{$j=2$}&\multicolumn{1}{c}{$j=3$}&\\ \hline
	    
	    0.2	& -1.9807& -1.9779&-1.9771&-1.9285&-1.9294\\ 
0.4	&-1.9285& -1.9187& -1.9156&-1.8005&-1.8016\\ 
0.6	&-1.8554& -1.8363& -1.8303&-1.6712&-1.6722\\ 
0.8	&-1.7729& -1.7441&-1.7351&-1.5542&-1.5550\\ 
1.0&-1.6898& -1.6520&-1.6404&-1.4515&-1.4522\\ 

1.4&-1.5425& -1.4897& -1.4742&-1.2840&-1.2846\\
1.8	& -1.4328& -1.3674& -1.3487&-1.1556&-1.1565\\
2.2	&-1.3592& -1.2812& -1.2590&-1.0551&-1.0567\\
2.6&-1.3129&-1.2219& -1.1958&-0.9751&-0.9779\\
3.0& -1.2853&-1.1817& -1.1515&-0.9103&-0.9152\\
4.0&-1.2576&-1.1313&-1.0934&-0.7941&-0.8084\\
6.0&-1.2503&-1.1122&-1.0678&-0.6678&-0.7145\\ \hline
    \end{tabular}
    \label{mytab}
  \end{center}
\end{table}
To show that Theorem \ref{lb} is applicable to more complicated systems, the lower bounds to the first eigenvalue for $n=3$, $Z_1=Z_2=Z_3=1$ (i.e. for $\mathrm{H}_3^{2+}$) are listed in Table \ref{mytab2}. Although  we can calculate for any positions of $x_1,x_2,x_3$, the values only for equilateral triangles are shown. $R$ indicates the length of the sides of the triangles.
 \begin{center}
     \setlength{\tabcolsep}{5pt}
	\tablefirsthead{\hline
	\multicolumn{1}{c}{\multirow{2}{*}{R}}&\multicolumn{3}{c}{\begin{tabular}{c}Lower  Bounds\\  with $\lambda\in[\lambda_j,\lambda_{j+1})$\end{tabular}}&\multicolumn{1}{c}{\multirow{2}{*}{$\mu_{1UB}$}}\\
	    &\multicolumn{1}{c}{$j=1$}&\multicolumn{1}{c}{$j=2$}&\multicolumn{1}{c}{$j=3$}\\ \hline}
	\tablehead{\hline}
	\tabletail{\hline}
	\topcaption{The first eigenvalue bounds for $\mathrm{H}_3^{2+}$}
	\tablelasttail{\hline}
	\begin{supertabular}{ccccc} 
0.2	& -4.3739& -4.3565&-4.3509&-4.1169\\
0.4	& -4.0662& -4.0104&-3.9928&-3.5919\\ 
0.6& -3.6931& -3.5987&-3.5703&-3.1556\\ 
0.8	& -3.3359& -3.2135&-3.1786&-2.8094\\
1.0	& -3.0342& -2.8937&-2.8558&-2.5327\\
1.4	& -2.6238& -2.4521&-2.4089&-2.1215\\
1.8& -2.4138& -2.1911&-2.1355&-1.8323\\
2.2	& -2.3177& -2.0310&-1.9575&-1.6183\\
2.6& -2.2768& -1.9339&-1.8422&-1.4540\\
3.0&-2.2603& -1.8798&-1.7735&-1.3243\\
3.4&-2.2538& -1.8531&-1.7365&-1.2198\\
3.8&-2.2514& -1.8412&-1.7172&-1.1342\\
4.2&-2.2505& -1.8364&-1.7066&-1.0633\\ \hline
  \end{supertabular} 
  \label{mytab2}
  \end{center}

\section{Improvements of Lower bounds by symmetry and Temple's inequality}\label{fifthsec}
\subsection{Temple's inequality}

If we have a lower bound $\mu_{2LB}$ to the second eigenvalue of a lower semibounded self-adjoint operator $A$, then we can obtain a lower bound $\mu_{1LB}$ to the first eigenvalue $\mu_1$ of $A$ by Temple's inequality. Let $\psi\in D(A)$ satisfy
\begin{equation}\label{myeq5.1}
\langle \psi,A\psi\rangle<\mu_{2LB}.
\end{equation}
 Then the Temple's inequality
\begin{equation}\label{myeq5.1.1}
\mu_1\geq \langle \psi,A\psi\rangle-\frac{\langle \psi,A^2\psi\rangle-\langle \psi,A\psi\rangle^2}{\mu_{2LB}-\langle \psi,A\psi\rangle},
\end{equation}
holds (see e.g. \cite{Te,Ka,GHT,BG}). For $\psi$ satisfying \eqref{myeq5.1} to exists, $\mu_{2LB}$ must be greater than $\mu_1$. Moreover, since $\langle \psi,A^2\psi\rangle-\langle \psi,A\psi\rangle^2=\langle \psi,(A-\langle \psi,A\psi\rangle)^2\psi\rangle\geq 0$, the larger $\mu_{2LB}$ is, the larger the right-hand side of \eqref{myeq5.1.1} is. An efficient method to make $\mu_{2LB}$ large is to restrict the function space in which we seek eigenfunctions to symmetric function spaces with respect to some transformations.

\subsection{Symmetry-adapted eigenvalue problem}
For a one-electron Hamiltonian \eqref{myeq3.0} an important symmetry is the one with respect to orthogonal transformation in $\mathbb R^3$ which maps $\{ x_1,\dots, x_n\}$ onto itself. The set $G$ of all such orthogonal transformations $g$ compose a group called a point group (see e.g. \cite{Le}). For each $g$ we can define a unitary operator $U(g)$ by $(U(g)\psi)(x)=\psi(g^{-1}x),\ \psi\in L^2(\mathbb R^3)$. Let $\varphi$ be an eigenfunction of $A$.  Then $U(g)\varphi$ is also an eigenfunction associated with the same eigenvalue, for $U(g)$ commutes with $A$. Thus, the set of all eigenfunctions associated with the same eigenvalue of $A$ forms a representation space of $G$. Therefore, restricting function space $L^2(\mathbb R^3)$ to the set of functions associated with one irreducible representation we can exclude  eigenvalues of $A$ associated with the other representations. As a result, in such a restricted function space the eigenvalue $\mu_2$ would be greater than the one in $L^2(\mathbb R^3)$. Since the eigenfunction $\psi_1$ associated with the first eigenvalue $\mu_1$ is expected to satisfy $U(g)\psi_1=\psi_1$, we consider the eigenvalue problem of $A$ in the function space $X:=\{ \psi\in L^2(\mathbb R^3): U(g)\psi=\psi,\ \forall g\in G\}$ corresponding to trivial representation.

The transformation $U(g)$ induces a corresponding transformation $\mathcal U(g)$ by
$$(\mathcal U(g)\mathbf c)_i:=\sum_{j,l=1}^M\mathcal F_{ij}\langle \psi_j,U(g)\psi_l\rangle c_l$$
in the vector space $\mathbb C^M$ on which the matrix $\mathcal B$ of Theorem \ref{lb} is a linear transformation, where $\mathcal F$ is the inverse matrix of $[\langle \psi_i,\psi_j\rangle]$. The matrix $\mathcal U(g)$ is a representation of $G$ in $\mathbb C^M$. To obtain the lower bound to the second eigenvalue $\mu_2$ of $A$ in $X$, we need to find the second eigenvalue of $\mathcal B$ in $\hat X:=\{\mathbf{c}\in \mathbb C^M: \mathcal U(g)\mathbf c=\mathbf c,\ \forall g\in G\}$. Since we consider only a finite group $G$, by Maschke's theorem (see e.g. \cite{Cr}) the representation $\mathcal U(g)$ is a direct sum of irreducible representations. If $\rho(g)$ and $\tilde {\rho}(g)$ are two different irreducible representation of $G$, then it is well known that $\sum_{g\in G}\rho_{ij}(g)\tilde\rho_{kl}(g)=0$ for any element $\rho_{ij}(g)$ and $\rho_{kl}(g)$ of $\rho(g)$ and $\tilde\rho(g)$. Taking as $\rho (g)$ the trivial representation i.e., $\rho_{11}(g)=1,\ \forall g\in G$, one finds $\sum_{g\in G}\tilde\rho_{kl}(g)=0$ for any irreducible representation $\tilde \rho(g)$ different from the trivial representation. Thus we can see that
\begin{equation}\label{myeq5.2}
\mathcal P:=\frac{1}{\lvert G\rvert}\sum_{g\in G}\mathcal U(g),
\end{equation}
is the projection on $\hat X$ corresponding to the trivial representation and $\mathcal P\mathbf c=0$ for $\mathbf c$ associated with the other representation, where $\lvert G\rvert$ is the order of $G$. Therefore, if we consider the eigenvalue problem for $\mathcal B\mathcal P$ instead of $\mathcal B$, the lower bound to $\mu_2$ in $X$ is obtained.
\begin{rem}
It is stated in \cite[Lemma 3.2.2]{We} that $\mathcal P$ in \eqref{myeq5.2} is a projection onto $\hat X$ in the sense that $\mathcal P^2=\mathcal P$ and $\mathrm{Ran}\, P=\hat X$. However, it is not enough for our purpose, because $\mathrm{Ran}\, (1-\mathcal P)$ is not specified. In the above arguments we can see that $\mathrm{Ran}\, (1-\mathcal P)$ is the subspace of $\mathbb C^M$ corresponding to the other irreducible representation of $G$ than the trivial representation.
\end{rem}

\subsection{Application of symmetry and Temple's inequality}
We apply the method in the previous subsections to $\mathrm{H}_2^+$. The point group of $\mathrm{H}_2^+$ is $D_{\infty h}$ (see e.g. \cite{Le}). However, $D_{\infty h}$ contains an infinite subgroup: rotation about the axis passing through the two nuclei. Since we need a finite group, we consider the finite subgroup $D_{2 h}$ of $D_{\infty h}$. The elements of $D_{2 h}$ are identity, inversion through the center of the nuclei, rotation by $\pi$ radians about three axes, reflection in three planes. The lower bounds to the second eigenvalues $\mu_2$ of $A$ in $L^2(\mathbb R^3)$ and $X$ are listed in Table \ref{mytab3}. The real number $\lambda$ in the definition of $\mathcal B$ is taken from $[\lambda_2,\lambda_3)$.

  \begin{center}
     \setlength{\tabcolsep}{5pt}
   \tablefirsthead{\hline
    \multicolumn{1}{c}{\multirow{2}{*}{R}}&\multicolumn{1}{c}{\begin{tabular}{c}$\mu_{2LB}$ \\in $L^2(\mathbb R^3)$\end{tabular}}&\multicolumn{1}{c}{\begin{tabular}{c}$\mu_{2LB}$ \\in $X$\end{tabular}}\\ \hline}
    \tablehead{\hline}
	\tabletail{\hline}
	\topcaption{The second eigenvalue bounds for $\mathrm{H}_2^{+}$}
	\tablelasttail{\hline}
	\begin{supertabular}{ccc}
0.2	& -0.5042& -0.4991\\
0.4	& -0.5199&-0.4964\\
0.6& -0.5563& -0.4923\\
0.8	& -0.6181& -0.4869\\
1.0	& -0.6954& -0.4807\\
1.4	& -0.8465& -0.4674\\
1.8& -0.9597& -0.4555\\
2.2	& -1.0318& -0.4469\\
2.6& -1.0731& -0.4413\\
3.0&-1.0947& -0.4373\\
4.0&-1.1107& -0.4269\\
6.0&-1.1114& -0.3971\\ \hline
 \end{supertabular} 
    \label{mytab3}
  \end{center}

As expected, the lower bounds $\mu_{2LB}$ to eigenvalues of $A$ in $X$ are greater than the ones in $L^2(\mathbb R^3)$ and therefore, they are lower bounds to eigenvalues corresponding to different irreducible representation of $D_{2 h}$. We also confirmed that the first eigenvalue of $\mathcal B\mathcal P$ is the same as the one of $\mathcal B$ which implies that the eigenfunction associated with the first eigenvalue belongs to the trivial representation. The lower bounds $\mu_{2LB}$ in $X$ are greater than $\mu_{1UB}$ in Table \ref{mytab}, though the ones in $L^2(\mathbb R^3)$ are not for large $R$. Thus we can use them in Temple's inequality \eqref{myeq5.1.1}. The lower bounds $\mu_{1LB}$ in Table \ref{mytab} are obtained by this way using function $\psi$ of the form \eqref{myeq3.2} and optimizing the parameters $\alpha$ and $\beta$.

\begin{rem}
\cite{GHT} applied Temple's inequality to the lower bounds to the first eigenvalues of $\mathrm{H}_2^+$ with $\mu_{2LB}=-0.5$ which equals to the second eigenvalue $\mu_2$ of $A$ in $X$ for $R=0$ and smaller than $\mu_2$ for $R>0$, because $\mu_2$ is increasing as a function of $R$.
\end{rem}

\end{document}